\documentclass[12pt,draft]{article}
\usepackage[T1]{fontenc}
\usepackage{array}
\usepackage{graphics}
\usepackage{styles/authblk}
\usepackage{styles/algorithm}
\usepackage{styles/algpseudocode}
\renewcommand{\algorithmicrequire}{\textbf{Input:}}
\renewcommand{\algorithmicensure}{\textbf{Output:}}
\algnewcommand\True{\textbf{true}\space}
\algnewcommand\False{\textbf{false}\space}
\algnewcommand\Or{\textbf{Or}\space}
\algnewcommand\And{\textbf{And}\space}
\algnewcommand\Xor{\textbf{Xor}\space}
\algnewcommand\To{\textbf{to}\space}
\algnewcommand\NewLineComment[1]{\Statex $\triangleright$ \footnotesize\textbf{#1}}

\algrenewcommand{\algorithmiccomment}[1]{\space \space *** \footnotesize{#1} ***}

\usepackage[dvips]{graphicx}

\usepackage{tikz}
\usepackage{}
\usetikzlibrary{trees, arrows, shapes}

\tikzset{
  treenode/.style = {align=center, inner sep=2pt, text centered,
    font=\sffamily},
  arn_r/.style = {treenode, circle, black, font=\sffamily\bfseries, draw=black,
    text width=1.5em},
  arn_n/.style = {treenode, rectangle split, rectangle split parts =2, draw=black,  rectangle split part fill={white, black!20},
    text width=1.5em, very thick},
  every edge/.append style={anchor=south,auto=falseanchor=south,auto=false,font=2.5 em},
}

\usepackage{amssymb}
\usepackage{amsmath}
\usepackage{amsthm}

\theoremstyle{plain}
\newtheorem{thm}{Theorem}[section]

\newtheorem*{cor}{Corollary}
\newtheorem{rem}{Remark}[section]
\theoremstyle{definition}
\newtheorem{defn}{Definition}[section]

\newtheorem{exmp}{Example}
\newtheorem{prob}{Problem}
\theoremstyle{remark}

\usepackage[square, numbers, comma, sort&compress]{natbib} 
\begin{document}

\title{Linear Algorithm for Conservative Degenerate Pattern Matching}
\author[1]{Maxime Crochemore}
\author[1]{Costas S. Iliopoulos}
\author[1]{Ritu Kundu}
\author[1]{Manal Mohamed}
\author[1]{Fatima Vayani}
\affil[1]{Department of Informatics, King's College London}

\renewcommand\Authands{ and }

\maketitle
\begin{abstract}
A \emph{degenerate symbol} $\tilde{x}$ over an alphabet $\Sigma$ is a non-empty subset of $\Sigma$, and a sequence of such symbols is a \emph{degenerate string}. A degenerate string is said to be \emph{conservative} if its number of non-solid symbols is upper-bounded by a fixed positive constant $k$. We consider here the matching problem of conservative degenerate strings and present the first linear-time algorithm that can find, for given degenerate strings $\tilde{P}$ and $\tilde{T}$ of total length $n$ containing $k$ non-solid symbols in total, the occurrences of $\tilde{P}$ in $\tilde{T}$ in $O(nk)$ time.
\end{abstract}

\section{Introduction}
\emph{Degenerate}, or \emph{indeterminate}, strings are found in Biology, Musicology and Cryptography.
They are defined by the occurrence of one or more positions which are represented by sets of symbols.
In \emph{conservative degenerate strings}, the number of such occurrences is bounded by $k$.
In music, single notes may match chords.
In encrypted and biological sequences, a position in one string may match exactly with various symbols in other strings.
\\

Previous algorithmic research of degenerate strings has been focused on pattern matching.
Pattern matching in degenerate strings is particularly relevant in the context of coding biological sequences.
Due to the degeneracy of the genetic code, two dissimilar DNA sequences can be translated into two identical protein sequences.
Without taking this degeneracy into account, many associations between biological entities can be overlooked.
For example, the following six DNA codons are all translated into the amino acid Leucine: $TTA$, $TTG$, $CTT$, $CTC$, $CTA$ and $CTG$.
This example highlights the significance of solving problems relating to degeneracy in strings.
In fact, special symbols to represent sets of DNA symbols have long been established by the IUPAC-IUBMB Biochemical Nomenclature Committee \cite{IUPAC-1985}.
For example, $R$ represents any purine ($A$ or $G$), $Y$ represents any pyrimidine ($C$, $T$ or $U$) and $N$ represents any nucleic acid.
An example of practical implications of such research is in the design of primers for cloning DNA sequences using PCR (Polymerase Chain Reaction).
Degenerate primers are used when their design is based on protein sequences,
which can be reverse-translated to $n^k$ different sequences, where $n$ is the length of the sequence.
\\

This paper introduces an algorithm which is a significant improvement from those published previously.
The first significant contribution for the problem of pattern matching of degenerate strings was in 1974 \cite{fischer1974string},
and was later improved \cite{muthukrishnan1994non}.
Later still, faster algorithms for the same problem were proposed \cite{indyk1998faster, kalai2002efficient}.
Since, many practical methods have been suggested \cite{holub2008fast, smyth2009adaptive, IMR06-WALCOM-CR},
as well as variations of the problem considered.
For example, a non-practical generalised string matching algorithm was introduced by Abrahamson in 1987 \cite{DBLP:journals/siamcomp/Abrahamson87}.
Most recently, Crochemore \textit{et al.} \cite{Isaac2014} reported an algorithm to find the shortest solid cover in a degenerate string with time complexity $O(2^k)$.
We report here a major improvement in time: $O(kn)$.
Further to the problem of pattern matching, the linear algorithm reported here can be applied to many different problems,
including finding cover and prefix arrays.
\\

The rest of the paper is organised in the following format:
The next section introduces the vocabulary and the notions that will be used in this paper. Section \ref{sec:con-deg-str-matching} formally defines the problem and presents the algorithm we have proposed. The algorithm is analysed in Section \ref{sec:analysis} and finally, Setion \ref{sec:conclusion} concludes the paper.

\section{Preliminaries}
\label{sec:prelim}
To provide an overview of our results we begin with a few definitions, generally following \cite{IMR06-WALCOM-CR,Isaac2014}. An \emph{alphabet} $\Sigma$ is a non-empty finite set of symbols of size $\vert \Sigma \vert$. A \emph{string} over a given alphabet is a finite sequence of symbols. The \emph{length} of a string $x$ is denoted by $\vert x \vert$. The \emph{empty string} is denoted by $\varepsilon$. The set of all strings over an alphabet $\Sigma$ (including empty string $\varepsilon$) is denoted by $\Sigma^*$.
\\

A \emph{degenerate symbol} $\tilde{x}$ over an alphabet $\Sigma$ is a non-empty subset of $\Sigma$, i.e., $\tilde{x} \subseteq \Sigma$ and $\tilde{x} \neq \emptyset$. $\vert \tilde{x} \vert$ denotes the size of the set and we have $1 \leq \vert \tilde{x} \vert \leq  \vert \Sigma \vert$.
A finite sequence $\tilde{X} = \tilde{x}_1\tilde{x}_2 \ldots \tilde{x}_n$ is said to be a \emph{degenerate string} if $\tilde{x}_i$ is a degenerate symbol for each $i$ from $1$ to $n$. In other words, a \emph{degenerate string} is built over the potential $2^{\vert \Sigma \vert} - 1$ non-empty sets of letters belonging to $\Sigma$. The number of the degenerate symbols, $n$ here, in a degenerate string $\tilde{X}$ is its \emph{length}, denoted as $\vert \tilde{X} \vert$. For example, $\tilde{X} = [a,b] [a] [c] [b, c] [a] [a, b, c]$ is a degenerate string of length $6$ over $\Sigma = [a, b, c]$.
If $\vert \tilde{x}_i \vert = 1$, that is, $\tilde{x}_i$ represents a single symbol of $\Sigma$, we say that $\tilde{x}_i$ is a \emph{solid} symbol and $i$ is a \emph{solid position}. Otherwise $\tilde{x}_i$  and $i$ are said to be \emph{non-solid symbol} and \emph{non-solid position} respectively. For convenience we often write $\tilde{x}_i = c$ ($c \in \Sigma$), instead of $\tilde{x}_i = [c]$, in case of solid symbols. Consequently, the degenerate string $\tilde{X}$ mentioned in the example previously will be written as $[a, b] a c [b, c] a [a, b, c]$. A string containing only solid symbols will be called a \emph{solid string}. Also as a convention, capital letters will be used to denote strings while small letters will be used for representing symbols. Furthermore, the degeneracy will be indicated by a \emph{tilde}, for example, $ \tilde{X} $ denotes a \emph{degenerate string} while  a plain letter like $X$ represents a \emph{solid string}.
The empty degenerate string is denoted by $\tilde{\varepsilon}$.
\\

A \emph{conservative} degenerate string is a degenerate string where its number of non-solid symbols is upper-bounded by a fixed positive constant $k$. The concatenation of degenerate strings $\tilde{X}$ and $\tilde{Y}$ is $\tilde{X}\tilde{Y}$. A degenerate string $\tilde{V}$ is a \emph{substring} (resp. \emph{prefix}, \emph{suffix}) of a degenerate string $\tilde{X}$ if $\tilde{X} = \tilde{U} \tilde{V} \tilde{W}$ (resp. $\tilde{X} = \tilde{V} \tilde{W}$, $\tilde{X} = \tilde{U} \tilde{V}$) for some degenerate strings $\tilde{U}$ and $\tilde{W}$. By $\tilde{X}[i..j]$, we represent a substring $\tilde{x}_i\tilde{x}_{i+1} \ldots \tilde{x}_j$ of $\tilde{x}$.
\\

For degenerate strings, the notion of symbol equality is extended to single-symbol \emph{match} between two degenerate symbols in the
following way. Two degenerate symbols $\tilde{x}$ and $\tilde{y}$ are said to \emph{match} (represented as $\tilde{x} \approx \tilde{y}$) if $\tilde{x} \cap \tilde{y}  \neq \emptyset$. Extending this notion to degenerate strings, we say that two  degenerate strings  $\tilde{X}$ and $\tilde{Y}$ \emph{match} (denoted as $\tilde{X} \approx \tilde{Y}$ ) if $\vert \tilde{X} \vert = \vert \tilde{Y} \vert$ and corresponding symbols in $\tilde{X}$ and $\tilde{Y}$ match, i.e., for each $i = 1, \cdots , \vert \tilde{X} \vert$ we have $\tilde{x}_i \approx \tilde{y}_i$. Note that the relation $\approx$ is not transitive. A degenerate string $\tilde{X}$ is said to \emph{occur} at position $i$ in another degenerate (resp. solid) string $\tilde{Y}$ (resp. $Y$) if $\tilde{X} \approx \tilde{Y}[i..i+\vert \tilde{X}]\vert-1]$ (resp. $\tilde{X} \approx Y[i..i+\vert \tilde{X}]\vert-1]$).

\section{Conservative Degenerate String Matching}
\label{sec:con-deg-str-matching}

\begin{prob}
\label{prob1}
Given a conservative degenerate pattern $\tilde{P}$ with $k$ non-solid symbols, and a solid text $T$, find all positions in $T$ at which $\tilde{P}$ occurs.
\end{prob}

\begin{exmp}
\label{example1}
We consider a degenerate pattern, $\tilde{P} = a[bc]da[bd]$ with $k=2$ and a text, $T = dacdabdadcabdac$ . 
Table \ref{table:example1} shows that $\tilde{P}$ occurs in $T$ at positions $2$ and $5$.
\begin{table}[!h]
\caption{Occurrence of $\tilde{P}$ in $T$}
\label{table:example1}
\centering
\resizebox{\columnwidth}{!}{%
\begin{tabular}{|c|*{15}{c}|}
\firsthline
  $i$    & 1 & 2 & 3 & 4 & 5 & 6 & 7 & 8 & 9 & 10 & 11 & 12 & 13 & 14 & 15\\
\hline
$t$ & $d$ & $a$ & $c$ & $d$ & $a$ & $b$ & $d$ & $a$ & $d$ & $c$ & $a$ & $b$ & $d$ & $a$ & $c$ \\
\hline
Matches & & $a$ & $[bc]$ & $d$ & $a$ & $[bd]$ & & & & & & & & & \\
        & & & &  & $a$ & $[bc]$ & $d$ & $a$ & $[bd]$ & & & & & & \\
\lasthline
\end{tabular}
}
\end{table}

\end{exmp}

For convenience, we compute a table $Pre[k, \vert \Sigma \vert]$ such that for each non-solid position $i$ ($1 \leq i \leq k$) and each letter $a \in \Sigma$, we have $Pre[i,a] = 1$ if $a \in \tilde{P}[i]$ and $0$ otherwise. After such $O(k\vert\Sigma\vert)$-time preprocessing, we can check in $O(1)$ time whether a non-solid position in $\tilde{P}$ matches a position in $T$ or not.

\subsection*{An Outline of Our Approach}
Our algorithm to solve Problem \ref{prob1} is built on the top of an adapted version of the sequential algorithm presented by Landau and Vishkin to find all occurrences of a (solid) pattern $P$ of length $m$ in a (solid) text $T$ of length $n$ with at most $e$ differences each \cite{Landau:1989:FPS:67205.67206}, where a difference can be due to either a mismatch between the corresponding characters of the text and the pattern, or a superfluous character in the text, or a superfluous character in the pattern. The modification required for our strategy is to treat only mismatches as the differences in Landau and Vishkin's algorithm. On the lines of the original Landau and Vishkin's algorithm, the modified one works in the following two steps .

\begin{description}
\item{Step 1:} Compute the suffix tree of the string obtained after concatenating the text, the pattern and a character $\#$ which is not present in $\Sigma \cup \Lambda$, i.e. $TP\#$; using the serial algorithm of Weiner \cite{Weiner:1973:LPM:1441424.1441766}.
\item{Step 2:} Let $Mismatch_{i, j}$ be the position in the pattern at which we have $j^{th}$ mismatch (when defined) between $T[i+1 .. i+m]$ and $P[1 .. m]$. In other words, $Mismatch_{i, j} = f$ represents $j^{th}$ mismatch from left to right and implies that $t_{i+f} \neq p_f$. In this step, we find $Mismatch_{i, j}$ for each $i$ and $j$ such that $0 \leq i \leq n-m$ and $1 \leq j \leq c+1$ where $c$ denotes the maximum of the two : $e$ and the total number of mismatches between $T[i+1 .. i+m]$ and $P[1 .. m]$. If some $Mismatch_{i, j} = m+1$, it signifies that there is an occurrence of the pattern in the text, starting at $t[i+1]$, with at most $e$ mismatches. $Mismatch_{i, j}$ can be computed from $Mismatch_{i, j-1}$ as follows :\\
Let $LCA_{si, sj}$ be the lowest common ancestor (in short LCA) of the leaves of the suffixes $T[si+1, n]$ and $P[sj+1]$ in the suffix tree and  $\vert LCA_{si, sj} \vert$ denotes its length.  $Mismatch_{i, j-1} =f$ implies that $T[i+1 .. i+f]$ and $P[1 .. f]$ is matched with $j-1$ mismatches. We want to find the largest $q$ such that $T[i+f+1 .. i+f+q] = P[f+1 .. f+q]$ and $t_{i+q+1} \neq p_{q+1}$, so that $Mismatch_{i, j} = q+1$. The desired $q$ is same as length of $LCA_{i+f, f}$. Thus, $Mismatch_{i, j} = f + \vert LCA_{i+f, f} \vert$.

\end{description}

Pseudocode for our approach is given as Algorithm \ref{algo:string-match}. It works in the following three stages :

\subsubsection*{\textsc{Stage 1}: Substitute} 
In the first stage, each of the non-solid symbols occurring in the given degenerate pattern is replaced by a unique symbol which is not present in $\Sigma$.
$\Lambda$ represents the set of these unique symbols i.e. $\{ \lambda_{i}\}$ such that $0 < i \leq k$. It is to be noted that the pattern, $p_{\lambda}$,  obtained by such a substitution will be a solid string. For example, 
$P_{\lambda}$ obtained from $\tilde{P}$ in Example \ref{example1} is given in Table \ref{table:example1_Stage1}.
\begin{table}
\caption{[\textsc{Stage 1}: Substitute]\space $P_{\lambda}$ obtained from $\tilde{P}$}
\label{table:example1_Stage1}
\centering
\resizebox{!}{!}{%
\begin{tabular}{|l|ccccc|}
\firsthline
 $\tilde{P}$ & $a$ & $[bc]$ & $d$ & $a$ & $[bd]$\\
\hline
$P_{\lambda}$ &  $a$ & $\lambda_1$ & $d$ & $a$ & $\lambda_2$\\
\lasthline
\end{tabular}
}
\end{table}

\begin{defn}
We define \emph{$\lambda$ positions} as the positions in $P_{\lambda}$ which contain $\{ \lambda_{i}\} \in \Lambda$. Note that these are same as the non-solid positions in $\tilde{P}$.
\end{defn}
\subsubsection*{\textsc{Stage 2}: Approximate Pattern Search}
The next stage comprises of using modified Landau and Vishkin's algorithm to search pattern $P_{\lambda}$ (solid) in text $T$ (solid) with at most $k$ mismatches in each occurrence. First, a suffix tree for the (solid) string $TP_{\lambda}$ is constructed. Then, LCA queries on this suffix tree are used to compute $Mismatch_{i, j}$ for each $i$ and $j$ such that $0 \leq i \leq n-m$ and $1 \leq j \leq k+1$. As explained in Remark \ref{remark: kLambdaMismatch}, $j$ will vary up to $k+1$ in $P_{\lambda}$'s case. Every $i$, such that $Mismatch_{i, k+1} = m+1$, marks the beginning of an occurrence of $P_{\lambda}$ in $T$ (at $i+1$) and thus added to the set $ApproximateMatch$.\\

Figure \ref{Fig: suffixTree} demonstrates the suffix tree for the string obtained from concatenating $T$ from Example \ref{example1} and $P_{\lambda}$ from the previous step, i.e $TP_{\lambda}$ which is $dacdabdadcabdaca\lambda_1da\lambda_2\#$. Note that each node of the suffix tree is stored as a pair (start, length) that represents the contiguous substring $S[start+1 .. start+length]$. In addition, a leaf node indicates the suffix it represents. A leaf node showing $i$ denotes a suffix $S[i+1 .. \vert S \vert]$.
Table \ref{table:example1_Stage2} shows the resultant $Mismatch[0 .. n-m, 1 .. k+1]$ array . This table provides the positions in $P_{\lambda}$ where it mismatches with the corresponding character in $T$. For example, $Mismatch[7,1] = 2$ denotes that the first mismatch between $T[8, 12]$ and $P_{\lambda}$ occurs at position $2$ in  $P_{\lambda}$ and rightly so as $t[8+2] = t[10]$ (i.e. $c$) does not match with $p_{\lambda}[2]$ (i.e. $\lambda_1$). As $P_{\lambda}$ occurs in $T$ with at most $2$ mismatches at locations $2, 5$ and $11$ (rows $1,4$ and $10$ contain $6$, i.e. $m+1$), $ApproximateMatch = [1, 4, 10]$.
\begin{rem}
\label{remark: kLambdaMismatch}
There will always be a mismatch between $P_{\lambda}$ and $T$ at $\lambda$ positions as each of the  $\lambda_{i} \in \Lambda$ is unique and does not occur in $\Sigma$ and hence in $T$. As there are $k$  $\lambda$ positions, at least $k$ mismatches are bound to be there for each position $i$ in the text starting at which the pattern is being matched against. More explicitly, each occurrence recorded in $ApproximateMatch$ has $k$ mismatches exactly.
\end{rem}
\begin{figure}[!h]
\caption{\textsc{[Stage 2: Approximate Pattern Match]} Suffix Tree for $TP_{\lambda}\#$}
\label{Fig: suffixTree}
\resizebox{!}{0.45\totalheight}{
\begin{tikzpicture}[->,>=stealth',level/.style={sibling distance = 4cm/#1,
  level distance = 4cm}] 
\node [arn_r] {\scriptsize{root}}
child{node [arn_r]{0, 1}
  child{node [arn_r]{1, 1}    
    child{node [arn_n]{5, 16 \nodepart{second} 3}  edge from parent node[sloped, above]{$\#...dadb$}
    }
    child{node [arn_n]{8, 13 \nodepart{second} 6} edge from parent node[sloped, above]{$\#...bacd$}
    }
    child{node [arn_r]{2, 1}
      child{node [arn_n]{3, 18 \nodepart{second} 0} edge from parent node[sloped, above]{$\#...dbad$}
      }
      child{node [arn_n]{15, 6 \nodepart{second} 12} edge from parent node[sloped, above]{$a\lambda_1da\lambda_2\#$}
      }
      edge from parent node[sloped, above]{$c$}
    }
    child{node [arn_n]{19, 2 \nodepart{second} 17} edge from parent node[sloped, above]{$\lambda_2\#$}
    }
    edge from parent node[sloped, above]{$a$}
  }
  child{node [arn_n]{9, 12 \nodepart{second} 8} edge from parent node[sloped, above]{$cabda...\#$}
  }
  edge from parent node[sloped, above]{$d$}
}
child{node [arn_n] {16, 5 \nodepart{second} 16} edge from parent node[sloped, above]{$\#\lambda_2ad\lambda_1$}
}	
child{node [arn_r]{1, 1}
  child{node [arn_r]{2, 1}
    child{node [arn_n]{3, 18 \nodepart{second} 1} edge from parent node[sloped, above]{$\#...dbad$}
    }
    child{node [arn_n]{15, 6 \nodepart{second} 13} edge from parent node[sloped, above]{$a\lambda_1da\lambda_2\#$}
    }
    edge from parent node[sloped, above]{$c$}
  }
  child{node [arn_n]{8, 13 \nodepart{second} 7} edge from parent node[sloped, above]{$\#...bacd$}
  }
  child{node [arn_n]{16, 5 \nodepart{second} 15} edge from parent node[sloped, above]{$\lambda_1da\lambda_2\#$}
  }
  child{node [arn_n]{19, 2 \nodepart{second} 18} edge from parent node[sloped, above]{$\lambda_2\#$}
  }
  child{node [arn_r]{5, 3}
    child{node [arn_n]{8, 13 \nodepart{second} 4} edge from parent node[sloped, above]{$\#...bacd$}
    }
    child{node [arn_n]{14, 7 \nodepart{second} 10} edge from parent node[sloped, above]{$ca\lambda_1da\lambda_2\#$}
    }
    edge from parent node[sloped, above]{$bda$}
  }
  edge from parent node[sloped, above]{$a$}
}
child{node [arn_n]{19, 2 \nodepart{second} 19} edge from parent node[sloped, above]{$\lambda_2\#$}
}
child{node [arn_r]{2, 1}
  child{node [arn_n]{3, 18 \nodepart{second} 2} edge from parent node[sloped, above]{$\#...dbad$}
  }
  child{node [arn_r]{10, 1}
    child{node [arn_n]{11, 10 \nodepart{second} 9} edge from parent node[sloped, above]{$\#...cadb$}
    }
    child{node [arn_n]{16, 5 \nodepart{second} 14} edge from parent node[sloped, above]{$\lambda_1da\lambda_2\#$}
    }
    edge from parent node[sloped, above]{$a$}
  }
  edge from parent node[sloped, above]{$c$}
}
child{node [arn_n]{20, 1 \nodepart{second} 20} edge from parent node[sloped, above]{$\#$}
}
child{node [arn_r]{5, 3}
  child{node [arn_n] {8, 13 \nodepart{second} 5} edge from parent node[sloped, above]{$\#...bacd$}
  }
  child{node [arn_n] {14, 7 \nodepart{second} 11} edge from parent node[sloped, above]{$ca\lambda_1da\lambda_2\#$}
  }
  edge from parent node[sloped, above]{$bda$}
}
; 
\end{tikzpicture}
}
\end{figure}
\begin{table}[!h]
\caption{[\textsc{Stage 2}: Approximate Pattern Search] \space $Mismatch$ array}
\label{table:example1_Stage2}
\centering
\resizebox{!}{!}{%
\begin{tabular}{|l|*{11}{c}|}
\firsthline
$j \downarrow$ $i\rightarrow$& $0$ & $1$ & $2$ & $3$ & $4$ & $5$ & $6$ & $7$ & $8$ & $9$ & $10$ \\ 
\hline
$1$ &$1$ & $2$ & $1$ & $1$ & $2$ & $1$ & $1$ & $2$ & $1$ & $1$ & $2$ \\ 
$2$ & $2$ & $5$ & $2$ & $2$ & $5$ & $2$ & $2$ & $3$ & $2$ & $2$ & $5$ \\
$3$ & $3$ & $6$ & $3$ & $3$ & $6$ & $3$ & $4$ & $5$ & $3$ & $3$ & $6$ \\ 
\lasthline
\end{tabular}
}
\end{table}

\subsubsection*{\textsc{Stage 3}: Filter} 
An occurrence in $ApproximateMatch$ reports a mismatch at a $\lambda$ position even if there is a match at the corresponding non-solid position in reality. For example, if some $\lambda_i$ has been substituted at a non-solid position containing, say $[b,c]$, and the corresponding symbol in $T$ is $c$, clearly it is a match but that position will be recorded as a `mismatch' in array $Mismatch$ because $\lambda_i$ does not match with $c$.
Thus, a mismatch of all the $k$ mismatches, found in an occurrence of solid $P_{\lambda}$ in $T$ identified by $ApproximateMatch$ in the preceeding step, can be seen as either \emph{real} or \emph{fake} when considered with respect to the match of the degenerate pattern $\tilde{P}$ and $T$. 
\begin{defn}
A mismatch at a position, say $e = Mismatch[i, j]$, is  \emph{real} if the corresponding symbols in the degenerate pattern $\tilde{P}$ and the text $T$ mismatch, i.e. $t[i+e] \not \approx \tilde{p}[e]$. Otherwise, the mismatch is \emph{fake}. 
\end{defn}

\begin{rem}
\label{remark: solidReal}
A mismatch at a solid position will always be real while one at a $\lambda$ position can either be real or fake.
\end{rem}

\begin{defn}
An \emph{approximate occurrence} is an occurrence of $P_{\lambda}$ in $T$ with $k$ mismatches whereas an occurrence of $\tilde{P}$ in $T$ with exact match is called an \emph{exact occurrence}.
\end{defn}

\begin{rem}
\label{remark: solidMismatch}
It follows from Remarks \ref{remark: kLambdaMismatch} and \ref{remark: solidReal} that if there is a mismatch even at a single solid position, total number of mismatches will exceed $k$ and such an occurrence will not figure as an approximate occurrence. Conversely, an appoximate occurence will have mismatches only at $\lambda$ positions. 
\end{rem}

For each location $i$ in the text where an \emph{approximate occurrence} of $P_{\lambda}$ has been found ($i \in ApproximateMatch$), each position of mismatch ($\lambda$ positions) in the pattern is checked for whether the mismatch is \emph{real} or not. If an approximate occurrence of pattern $P_{\lambda}$ in text $T$ contains a real mismatch, it can be observed that it cannot represent an exact occurrence of $\tilde{P}$ whereas the approximate occurrence containing only fake mismatches will be same as an exact occurrence.
The set of all such exact occurrences is the solution to our Problem \ref{prob1}. This step, therefore, filters out and discards the approximate occurrences with real errors.\\

Table \ref{table:example1_Stage3} elucidates this stage for the example being considered. With values given by $ApproximateMatch = [1, 4, 10]$ from the previous stage, we test each $\lambda$ position from $\Lambda = [2,5]$ to check if the mismatch is real or fake. At first $\lambda$ position (i.e. $2$), $t[1+2]$ (i.e. $c$) matches $\tilde{p}[2]$ (i.e. $[b,c]$), thus the mismatch is fake. The mismatch for the second $\lambda$ position (i.e. 5) is also fake owing to the fact that $t[6] \approx \tilde{p}[5]$. Therefore, location $2$ is recorded as an occurrence of exact match of  $\tilde{P}$ in $T$. Similar is the case of location $5$ (i.e. value $4$). But for value $10$, even if the first mismatch is fake ($t[12]$ (i.e. $b$) $\approx \tilde{p}[2]$ (i.e. $[b,c])$), the fact that $t[15]$ (i.e. $c$) $\not \approx \tilde{p}[5]$ (i.e. $[b,d]$) makes the second mismatch real. Therefore, location $10$ is discarded.
And thus the correct solution to Example \ref{example1} is obtained.

\begin{table}[!h]
\caption{[\textsc{Stage 3}: Filter] \space Checking Mismatches in  Approximate Occurrences of $\tilde{P}$ in $T$}
\label{table:example1_Stage3}
\centering
\resizebox{\columnwidth}{3cm}{%
\begin{tabular}{|c|*{15}{c}|}
\firsthline
  $i$    & 1 & 2 & 3 & 4 & 5 & 6 & 7 & 8 & 9 & 10 & 11 & 12 & 13 & 14 & 15\\
\hline
$T$ & $d$ & $a$ & $c$ & $d$ & $a$ & $b$ & $d$ & $a$ & $d$ & $c$ & $a$ & $b$ & $d$ & $a$ & $c$ \\
\hline
\hline
Approximate & & $a$ & $\lambda_1$ & $d$ & $a$ & $\lambda_2$ & & & & & & & & & \\
Ocuurences  & & & $\downarrow$ & & & $\downarrow$ & & & & & & & & & \\ 
          & & & $[bc]$ & & & $[bd]$ & & & & & & & & & \\ 
          & & & $\Downarrow$ & & & $\Downarrow$ & & & & & & & & & \\  
          & & & Fake & & & Fake & & & & & & & & & \\
\cline{2-16}
          & & & &  & $a$ & $\lambda_1$ & $d$ & $a$ & $\lambda_2$ & & & & & & \\
          & & & &  &  & $\downarrow$ &  &  & $\downarrow$ & & & & & & \\
          & & & &  &  & $[bc]$ &  &  & $[bd]$ & & & & & & \\
          & & & &  &  & $\Downarrow$ &  &  & $\Downarrow$ & & & & & & \\
          & & & &  &  & Fake &  &  & Fake & & & & & & \\
\cline{2-16}
          & & & &  & & & & & & &$a$ & $\lambda_1$ & $d$ & $a$ & $\lambda_2$  \\
          & & & &  & & & & & & & & $\downarrow$ & & & $\downarrow$  \\
          & & & &  & & & & & & & & $[bc]$ & & & $[bd]$  \\
          & & & &  & & & & & & & & $\Downarrow$ & & & $\Downarrow$  \\
          & & & &  & & & & & & & & Fake & & & Real  \\

\lasthline
\end{tabular}
}
\end{table}

\begin{algorithm}[!h]
\caption{Conservative Degenerate String Matching Algorithm}
\label{algo:string-match}
\begin{algorithmic}[1]
\Require Pattern $\tilde{P}$ of length $m$,
\Statex \hskip \algorithmicindent  ~ Text $T$ of length $n$, 
\Statex \hskip \algorithmicindent   ~ Number of non-solid symbols $k$
\Ensure The set of indices of $T$  where $\tilde{P}$ occurs in $T$ 
\Statex
\NewLineComment{Substitute:}
\State $\Lambda \gets \{ \lambda_{i} | \, \lambda_i \not \in \Sigma \; and \; 0 < i \leq k\}$ 
  \State $P_{\lambda} \gets $string obtained after substituting $i^{th}$ non-solid symbol in $\tilde{P}$ with $ \lambda_{i}$ in $\Lambda$ $\forall$ $i$ such that $ 0 < i \leq k$  
\Statex
\NewLineComment{Approximate Pattern Search:}
  \State Build Suffix Tree for the string $TP_{\lambda}\#$
  \State $ApproximateMatch \gets \emptyset$
  \For{$i \gets 0$ \To $n-m$} \label{state: search}
       \State $f \gets 0$
       \For{$j \gets 1$  \To $k+1$} 
            \State $Mismatch[i, j] = f + \vert LCA_{i+f, f} \vert$ 
            \State $f \gets Mismatch[i, j]$
      \EndFor
      \If{$Mismatch[i, k+1] = m+1$} \Comment{approximate occurrence found}
        \State{Add $i$ to $ApproximateMatch$}
      \EndIf
  \EndFor \label{state: searchEnd}
\Statex
\NewLineComment{Filter:}
  \State $Occ \gets \emptyset$
  \For{each $i \in ApproximateMatch$}
      \State $flagAllFake \gets \True$
      \For{each $e \in \Lambda$}
          \If{$t[i+e] \not \approx \tilde{p}[e]$}
            \State $flagAllFake \gets \False$
            \State \textbf{Break} \Comment{real mismatch}
          \EndIf
      \EndFor
      \If{$flagAllFake$} \Comment{all fake mismatches}
        \State{Add $i+1$ to $Occ$} \Comment{exact occurrence found}
      \EndIf
  \EndFor
\State \Return{$Occ$}

\end{algorithmic}
\end{algorithm}

\section{Agorithm Analysis}
\label{sec:analysis}
\begin{thm}
Algorithm \ref{algo:string-match} correctly computes all occurrences of $\tilde{P}$ in $T$ in $O(kn)$ time complexity.
\end{thm}
\begin{proof}
Landau and Vishkin's algorithm correctly finds all occurrences of $P_{\lambda}$ in $T$ with at most $k$ mismatches in  $O(kn)$ time for a fixed alphabet. $P_{\lambda}$ differs from  $\tilde{P}$ only at the $\lambda$ positions which are equal to $k$ in number. In addition, each of the $\lambda$ positions causes a mismatch. Notably, an exact occurrence of $\tilde{P}$ in $T$ will be given by an approximate occurrence of $P_{\lambda}$ in $T$ with mismatches only at $\lambda$ positions and all of these mismatches must be fake. All such occurrences where mismatches occur only at $k$ $\lambda$ positions are guaranteed to be captured by the approximate occurrences given in $ApproximateMatch$. Also, as a consequence of Remark \ref{remark: solidMismatch}, an approximate occurrence (for which number of mismatches are at most $k$)  will never have a mismatch at any solid position. The filtering stage checks each of the mismatches in an approximate occurrence and if all of these mismatches are found to be fake, we have an exact occurrence.  Thus, at the end of the filtering stage, we have all the occurrences of an exact match only.\\ 

The substitution stage can be performed in $O(n)$ time. As mentioned previously, the approximate pattern-search stage using modified Landau and Vishkin's algorithm computes $ApproximateMatch$ in $O(kn))$ time for a fixed sized alphabet as the suffix tree is constructed in linear time with respect to the size of the input string ($n+m$) and computation of $Mismatch$ array (lines \ref{state: search} to \ref{state: searchEnd}) takes $O(kn)$ time. The filtering stage, in the worst case ($ApproximateMatch$ contains $0$ to $n-m$), needs to process each location in $T$ and to check whether mismatch at every $\lambda$ position is real or fake. This check can be performed in constant time after $O(k\vert \Sigma \vert)$-time pre-processing as mentioned earlier, which yields $O(kn)$ time requirements for this stage. Thus, in $O(k\vert \Sigma \vert + n + kn + kn)) = O(kn)$ time Algorithm \ref{algo:string-match} correctly computes all occurrences of $\tilde{P}$ in $T$.
\end{proof}

\begin{cor}
Given degenerate strings $\tilde{P}$ and $\tilde{T}$ of total length $n$ containing $k$ non-solid symbols in total, one can
compute occurrences of $\tilde{P}$ in $\tilde{T}$ in $O(nk)$ time.
\end{cor}


\section{Conclusion}
\label{sec:conclusion}
In this paper, we studied the matching problem of conservative degenerate strings and presented an efficient algorithm that can find, for given degenerate strings $\tilde{P}$ and $\tilde{T}$ of total length $n$ containing $k$ non-solid symbols in total, the ocurrences of $\tilde{P}$ in $\tilde{T}$ in $O(nk)$ time, i.e. linear to the size of the input. In particular, we used the novel technique of substituting the non-solid symbols in the given degenerate strings with unique solid symbols, which let us make use of the efficient approximate pattern search solution for solid strings to get an efficient solution for  degenerate strings. It would be interesting to see how well the presented algorithm behaves in practice and to apply it to solve a vast number of problems like prefix/border array, suffix trees, covers, repetitions, seeds, decomposition etc.


\bibliographystyle{unsrtnat} 

\bibliography{Bibliography} 

\end{document}